\documentclass[a4paper]{article}
\usepackage[latin1]{inputenc} % ou \usepackage[utf8]{inputenc}
\usepackage[T1]{fontenc} % ou \usepackage[OT1]{fontenc}
\usepackage{RR}
\usepackage{hyperref}
\usepackage{subfigure}
\usepackage{amssymb}

\newtheorem{theorem}{Theorem}[section]
\newtheorem{lemma}[theorem]{Lemma}

\newenvironment{proof}[1][Proof]{\begin{trivlist}
\item[\hskip \labelsep {\bfseries #1}]}{\end{trivlist}}
\RRdate{May 2012}
\RRversion{2\thanks{Enhancement of the protocol description and extension of the simulation part}}
\RRdater{March  2013}

\RRauthor{% les auteurs
Alexandre Mouradian%\thanks{Footnote for first author}%
  \thanks[sfn]{University of Lyon, INRIA, France, INSA Lyon, CITI, F-69621, France}%
  \and
Isabelle Aug\'e-Blum
 \thanksref{sfn}
  \and
Fabrice Valois
\thanksref{sfn}
}
\authorhead{Mouradian, Aug\'e-Blum \& Fabrice Valois}
\RRtitle{RTXP : A Localized Real-Time Mac-Routing Protocol for Wireless Sensor Networks}
\RRetitle{RTXP : A Localized Real-Time Mac-Routing Protocol for Wireless Sensor Networks}
\titlehead{RTXP}
\RRresume{Les protocoles d\'evelopp\'es pour les r\'eseaux de capteurs sans fil sont con\c{c}us principalement pour r\'eduire la consommation d'\'energie et permettre au r\'eseau de s'auto-organiser. Les applications critiques temps-r\'eel n\'ecessitent, en plus de ces caract\'eristiques, que les communications soient fiables et que le d\'elai de bout en bout soit born\'e. Dans ce document nous nous int\'eressons principalement \`a des remont\'ees d'alarmes vers le puits. Nous proposons RTXP (Real-Time X-layer Protocol) un protocole de communication temps-r\'eel. Ce protocole d\'efinit les m\'ecanismes des niveaux MAC et routage. Le but est de garantir une contrainte sur le d\'elai de bout en bout, tout en conservant des bonnes performances pour les autres param\`etres tel que la consommation d'\'energie. Pour cela le protocole utilise un syst\`eme de coordonn\'ees virtuelles qui permet un acc\`es au medium et une s\'election du noeud relayeur d\'eterministes. Dans ce document nous d\'ecrivons les m\'ecanismes de ce protocole. Nous d\'eterminons des bornes th\'eoriques sur le d\'elai de bout en bout et sur la capacit\'e du protocole. Une campagne de simulation permet de confirmer les r\'esultats th\'eoriques et de comparer RTXP \`a une solution centralis\'ee. Des simulations avec des liens radios non fiables sont aussi effectu\'ees. Les liens radio introduisent un comportement probabiliste. Cependant, les r\'esultats montrent que RTXP \`a de meilleures performances qu'une solution non d\'eterministe. Cela d\'emontre l'utilit\'e des solutions d\'eterministes m\^eme quand les liens radio ne sont pas fiables.%classe les noeuds qui ont le m\^eme nombre de sauts depuis le puits. Cela sert \`a choisir le noeud qui doit relayer un paquet. Ce syst\`eme de coordonn\'ees permet aussi de diff\'erentier les noeuds d'un deux-voisinage ce qui autorise \`a avoir des acc\`es au m\'edium d\'eterministes. Dans ce document nous d\'ecrivons les m\'ecanismes de ce protocole, nous le comparons \`a un protocole existant. Une campagne de simulation permet de confirmer les r\'esultats th\'eoriques. Nous discutons les performances de ce protocole dans le cas o\`{u} les liens radio sont non fiables.}
}

\RRabstract{Protocols developed during the last years for Wireless\hspace{2pt} Sensor Networks (WSNs) are mainly focused on energy efficiency and autonomous mechanisms (e.g. self-organization, self-configuration, etc). Nevertheless, with new WSN applications, appear new QoS requirements such as time constraints. Real-time applications require the packets to be delivered before a known time bound which depends on the application requirements. We particularly focus on applications which consist in alarms sent to the sink node. We propose Real-Time X-layer Protocol (RTXP), a real-time communication protocol. To the best of our knowledge, RTXP is the first MAC and routing real-time communication protocol that is not centralized, but instead relies only on local information. The solution is cross-layer (X-layer) because it allows to control the delays due to MAC and Routing layers interactions. RTXP uses a suited hop-count-based Virtual Coordinate System which allows deterministic medium access and forwarder selection. In this paper we describe the protocol mechanisms. We give theoretical bound on the end-to-end delay and the capacity of the protocol. Intensive simulation results confirm the theoretical predictions and allow to compare with a real-time centralized solution. RTXP is also simulated under harsh radio channel, in this case the radio link introduces probabilistic behavior. Nevertheless, we show that RTXP it performs better than a non-deterministic solution. It thus advocates for the usefulness of designing real-time (deterministic) protocols even for highly unreliable networks such as WSNs.}
\RRmotcle{r\'eseaux de capteurs sans fil, syst\`eme de coordonn\'ees virtuelles, contraintes applicatives, MAC, routage}
\RRkeyword{wireless sensor networks, virtual coordinate system, application constraints, MAC, routing}
 \RRprojet{urbanet}  % cas d'un seul projet
\RRprojet{urbanet} % theme du projet Apics
 \RCGrenoble % Grenoble - Rh\^one-Alpes
\begin{document}
\RRNo{7978}
\makeRR   % cas d'un rapport de recherche
%% \makeRT % cas d'un rapport technique.
%% a partir d'ici, chacun fait comme il le souhaite
\section{Introduction}

A WSN is composed of nodes deployed in an area in order to monitor parameters of the environment. Those nodes are able to send information to dedicated nodes called sinks, without the need of a fixed network infrastructure, in a multi-hop fashion. Every node is able to forward messages from the other nodes. They usually run on batteries so they should consume as little energy as possible in order to increase the network lifetime. Because WSNs can contain thousands of nodes, the cost of a node should be as low as possible, this leads to design nodes with poor capabilities (computation, radio, memory, etc...). In the past few years WSNs have been a very active research field which has led to interesting contributions at all communication layers. This is due to the great expectations put in WSN applications. In fact, lots of applications have been proposed in the literature, such as volcano monitoring~\cite{tan10}, air pollution monitoring~\cite{khedo10}, landslide detection~\cite{ramesh09} and so on.

Due to the previously mentioned characteristics of WSNs, network protocols have been designed mainly in order to reduce energy consumption and to provide autonomous network mechanisms. Nevertheless some applications need more than these characteristics. Indeed, critical applications require more reliability  and the respect of time constraints. For instance the aforementioned landslide detection application should give guarantees on the delivery of alert messages. Protocols which can deliver messages with guaranteed end-to-end delay are called real-time protocols. They are usually classified into two categories, soft real-time and hard real-time: in the first case, some messages can miss the deadline with no consequences (video) while in the second case, the delay constraint should be always respected
whatever the circumstances are because of the possible impact on human life, on the environment or on the
financial cost. Due to the probabilistic nature of the radio links in WSNs, strict time constraint is not acheivable, thus the time bound must be associated to a given reliability. This parameter is thus a main concern in the design of a WSN real-time protocol.

Hard real-time constraints cannot be met with the current WSN protocols of the literature either because of their lack of determinism which implies unbounded delays or low reliability, or because they do not take into account the aforementioned characteristics of WSNs.

In this paper we propose, a new localized real-time cross-layer protocol, RTXP. This protocol aims at giving a bound on the end-to-end delay in a WSN. In order to handle real-time requirements, deterministic mechanisms must be introduced at MAC and routing layers. The interactions between these two layers must also be carefully controlled in order to avoid unexpected and unbounded delays. We thus claim that a cross-layer design where MAC and routing layers share information should be preferred. Our approach is to bound the duration of one hop \footnote{We define the duration of one hop to be the time needed for a node to access the medium and send a packet} and the number of hops to reach the sink. To avoid unbounded delays and unbounded route lengths, the access to the medium and the choice of the forwarder must be deterministic. Our approach is based on a suitable Virtual Coordinate System (VCS) \cite{Mouradian12}. This VCS allows the nodes to get information on their distance to the sink in number of hops. It also differentiates nodes having the same hop-counts in order to improve the forwarder selection. Finally, it gives a unique identifier to the nodes in a 2-hop neighborhood in order to deterministically access the medium. The VCS is constructed with local information (the neighbors of a node) and it is the only information used by RTXP. Our proposition is thus localized, no global view of the system is needed, the approach is therefore more scalable than centralized solutions.

Under harsh channel conditions, no hard real-time guarantee can be given whatever the protocol used, because a message may need an infinite number of retransmissions to be correctly transmitted (even if the probability of this event is low). Indeed, even if the protocol is deterministic, the radio link introduces probabilities. In the paper we show that a deterministic protocol allows to achieve better performances (notably reliability) than non real-time solutions.

In section 2 we discuss the advantages and drawbacks of existing WSNs MAC and routing protocols for real-time applications. In section 3 we introduce the hypothesis and the requirements of our solution.
In section 4 we present the details of our proposition, RTXP. In section 5 we give theoretical bound on the end-to-end delay and the capacity of the protocol. Section 6 presents simulation parameters and results, we compare RTXP to a centralized solution and a non real-time protocol. In section 7, we conclude on the protocol properties and performances and we present our future works.

%------------------------------------------------------------------------- 
\section{Related work}
WSN MAC and routing protocols were widely investigated during the last years. Unfortunately only few contributions were focused on real-time protocols. In this section we discuss the main results for MAC only, routing only and cross-layer protocols.

\subsection{Medium Access Control}
\label{mac}
To save energy, MAC protocols for WSNs usually use a duty cycle ~\cite{polastre04} mechanism to save energy. Since the receiving, sending and listening energy costs are approximately the same for usual radio chips~\cite{cc2500}, the only way to save energy is to turn off the radio (e.g. to switch to sleep mode). Duty-cycling consists in nodes alternately waking up and going to sleep mode.
MAC protocols can be classified into two main categories: synchronous and asynchronous. In synchronous protocols, the nodes know the schedules of the wakeup of other nodes (in their neighborhood or in the whole network). Usually a mechanism is used to synchronize the clocks of the nodes. They thus share a common global or local clock. In asynchronous protocols, the synchronization exists but it is event-based. A communicating node and its neighborhood synchronize only for the time of a communication but without exchanging the values of their clocks. In both cases, the medium can be accessed either by contention or in a deterministic way.
%todo: mettre des figures pour la synchro

BMAC~\cite{polastre04} illustrates a classical asynchronous protocol for WSNs. First, nodes pick a random wakeup time and then alternately sleep and wakeup. BMAC uses LPL (Low Power Listening) method which consists in listening periodically to the channel for radio activity. When a node needs to send a message, it sends a preamble (sequence of bits) which duration is equal to the duty-cycle period. Each node, when it wakes up, senses the channel, if it detects the preamble it stays awake until the end of the communication in order to receive the packet, otherwise it goes back to sleep. The main advantage of this protocol is that there is no need for a time synchronization algorithm. Nevertheless the length of the preamble and the overhearing problem (nodes listening the preamble even if they are not the destination of the packet) can leads to high energy consumption. Enhancements have been propose by XMAC~\cite{buettner06} and MFP~\cite{bachir06} for instance. They aim at reducing energy consumption by avoiding unnecessary listening.

SMAC~\cite{ye04} is a synchronous protocol. All the nodes in a neighborhood synchronize in order to wake up at the same time and access the medium by contention using carrier sense and RTS/CTS mechanisms. D-MAC~\cite{lu04} enhances this mechanism by scheduling the wakeup time in function of the distance to the sink. This reduces the end-to-end delay of the packets in the network.

%cooperative MAC with network coding may be able to enhance the reliability of the MAC for WSN.

In the previously cited protocols the channel access is not deterministic (i.e. collisions can occur). It implies that the time to access the medium is not guaranteed. This leads to unbounded delay to perform one hop which is not suitable for real-time applications.

Solutions that provide deterministic access to the medium have been proposed, such solutions allow to respect real-time constrains.

IEEE 802.15.4 \cite{802.15.4} is a standard which defines a physical and MAC layer for WSNs. Networks can be peer-to-peer or star networks. In each case, at least one node acts as a coordinator and sends synchronization beacons. Between beacons, a superframe is defined. It is composed of two parts, Contention Access Period (CAP) and Contention Free Period (CFP). In the CFP, the access are guaranteed allowing real-time communications. This feature is used in the ISA100.11a \cite{isa}, standard for wireless systems for industrial automation. 
Nevertheless, scalability \cite{yedavalli08} and reliability \cite{anastasi11} issues have been highlighted.
%Nevertheless, the hierarchical approach used in the standard  is not suited to very large scale networks (because maintaining the hierarchical structure is costly).

I-EDF~\cite{caccamo02} proposes an approach based on Early Deadline First scheduling algorithm. I-EDF is a synchronous protocol. The network is divided in hexagonal cells which use different radio frequencies. There are two kinds of communications, the intra-cell communications and the inter-cell communications. For intra-cell communications, all the nodes of a cell know the periodicity and length of the packets of every other node, in the same cell, so they can apply the EDF algorithm. Between two intra-cell communications there are slots reserved for inter-cell communications. There are 6 directions in which a node can emit (corresponding to the six edges of the hexagonal cell) each inter-cell slot is allocated to a direction (each router emitting on a different frequency which corresponds to the one of its neighbor in the current direction). This proposition allows real-time communications but does not take into account the energy consumption (there is no duty-cycle). Moreover the nodes must have the capacity of communicating on 7 channels and the hexagonal cell topology is restraining.

On the contrary f-MAC~\cite{roedig06} proposes a localized and asynchronous approach. The principle is that nodes periodically send small packets (called frames) with a dedicated period, each node in the neighborhood having a unique period attributed. The authors show that, by applying mathematical rules for the choice of  the periods, it can be guaranteed that a frame of each node will actually be transmitted without collision. This MAC guarantees hard real-time constraints on perfect radio links and the transmission mechanisms are very simple. Nevertheless it has very a poor channel utilization, a quite high energy consumption (no duty-cycle) and the maximum delay increases exponentially with the number of nodes in the same collision domain.

Dual-MAC \cite{watteyne06} is a hard real-time MAC for linear WSNs. Thanks to the linear topology of the WSN, the routing problem has not to be taken into account. At initialization the network is divided into cells. The protocol can operate in two modes. The ``unprotected mode'' allows a near-optimal transmission time when there is no collision. When a collision occurs, a ``protected mode'' is triggered in which a path is reserved in order to guarantee bounded transmission times. The furthest node from the sender is elected as the packet relay. The two modes allow the protocol to provide a good trade-off between average performances and respect of strict time constraints. The worst case is computed and the protocol is formally verified with a model checker. Nevertheless the solution is only suited for linear networks and does not take into account energy consumption issues in the design of the protocol.

The MAC protocols described in this paper do not allow to respect strict time constraints (hard real-time) while taking into account the previously cited requirements specific to WSNs. The propositions which allow to respect hard time constraints do not take into account energy consumptions issues, radio chip limitations, or are difficult to integrate with a routing layer. Other propositions are more suited to WSNs but do not allow to respect strict time constraints.

\subsection{Routing}
In this section we focus on routing protocols for WSNs. They can be classified into four categories: probabilistic, hierarchical, location-based and broadcast-based.

In probabilistic routing protocols, such as Random Walk Routing \cite{Servetto02}, forwarders are elected by making random choices. This class of protocols cannot be used for real-time communications because of its lack of determinism. Indeed, it leads to unbounded routes length which do not allow to provide a bound on end-to-end delay.

In hierarchical protocols nodes can be grouped into clusters (for example I-EDF \cite{caccamo02} uses a cluster hierarchy) or organized as trees. In the former case, a leader called cluster head is a node which is responsible for collecting the data sensed by the nodes of the cluster and sending it to the sink. In this case the length of the route is bounded, it could thus be used for real-time communications. In the latter case, a tree representing the network is constructed. RPL \cite{rpl} is a routing protocol for IPv6 in low-power and lossy networks standardized by the IETF. It uses Directed Acyclic Graph as routing structure along with several routing metrics. The main issue with hierarchical schemes, is that maintaining the structure can be expensive in terms of energy consumption in highly dynamic networks. Moreover, in the case of a node failure many nodes may result disconnected from the sink.

Location-based  protocols are making forwarding decisions depending on the geographic location of the destination of the packet. A method for choosing a forwarder is to elect the neighbor of the sender which is closer to the sink. This mechanism is called greedy forwarding and has good performances in WSNs~\cite{karp00}. The distance can be geographic distance~\cite{karp00} but \cite{watteyne07} shows that lack of accuracy of the geographic coordinates leads to bad performances. Furthermore GPS chips cost is high and the energy consumption is not negligible. VCSs have been developed in order to address those issues. \cite{rao03} and \cite{cao04} propose two different types of VCS, greedy routing is then used with the coordinates. With such a technique the number of hops is not known a priori. Nevertheless a bound on the number of hops could be deduced from the position of the furthest node from the sink. This solution is however only applicable in dense networks. Indeed, a hole in the network can induce a longer route than expected \cite{karp00}. This can be an issue in the case of real-time communications because we cannot bound the hop-count of a packet and thus the time it takes to reach the sink. SPEED~\cite{Stankovic03} is a routing protocol based on geographic coordinates. A node keeps a table of its neighbors with a metric that represents their speed. The speed of a neighbor is computed by dividing the advance in geographic distance it provides in direction of the destination, by the delay to forward the packet to that neighbor. The forwarder is selected if its speed allows to respect the deadline. SPEED does not bound the end-to-end delay. Nevertheless, it provides a congestion detection mechanisms. MMSPEED~\cite{Felemban06} increases the reliability of SPEED by using a multi-path scheme. RPAR \cite{chipara06} enhences SPEED and MMSPEED by taking into account energy consumption and lossy radio links.

In broadcast-based routing protocols, a node does not need to store explicit information on the network topology. It broadcasts the message and the choice of the next hop is done by nodes which receive it. The choice is based on a metric that can depend on the coordinates (geographic or virtual) and other parameters of the potential forwarders or it can be done randomly. GRAB \cite{ye04} can be classified in this category of routing protocols. GRAB introduces a cost-field which can be seen as a VCS or a metric, indeed the cost-field represents the cost for a node to reach the sink. In GRAB the hop-count is used as a cost-field. During the initialization phase each node is assigned its distance to the sink, in number of hops, as coordinate. Then packets are routed using gradient-routing which consists in choosing the forwarder which has the lowest cost-field value. As many nodes with the same hop-count can listen the packet, the selection of the forwarder can be based on a random value and multiple forwarders can be elected, this induces multiple paths. The advantages of such a solution are that the number of hops to reach the sink is known and multiple path leads to more reliability. Nevertheless GRAB does not give information on the physical organization of nodes with the same hop-count. This information could be useful in order for example to select the best forwarder for a packet in a deterministic way (in the case of greedy routing it is the nearest from the sink), it could be used as well to access the channel in a X-layer scheme. SGF \cite{Huang09} and LQER \cite{Chen08} propose similar schemes. In SGF only one node is chosen in an opportunistic manner. LQER add information on the link quality. Both solutions suffer from the same aforementioned drawbacks of GRAB.

Among the cited routing protocols, some take into account the time in order to route packets \cite{Stankovic03}, others may allow to bound the length of a route \cite{ye04} or to provide reliable end-to-end communication. Nevertheless, none is able to guarantee the respect of real-time constraints.

\subsection{Cross-layer}
\label{centralized}
Solutions which integrate both MAC and routing mechanisms have been proposed. These solutions allow to plan routes and medium access simultaneously.

PR-MAC~\cite{chen07} is a synchronous real-time MAC and routing protocol for WSNs. The aim is to detect events and then to set up a periodical monitoring of the area where the event occurred. When an event is sensed in the network an alarm is sent to the sink. The sink responds with a packet that reserves a path for the periodical monitoring. The nodes on the path then wake up two times, once for the traffic from the source to the sink and once for the traffic from the sink to the source. The path is reserved with a given radio frequency. Once the path is reserved the monitoring packets are transmitted in real-time but the reservation phase is non real-time and induces an overhead. Moreover the protocol assumes the radio handles multi-channel communications.

TSMP~\cite{pister08} uses a MF-TDMA scheme to access the medium. It uses a centralized scheduling, where time-slots and channels are assigned to nodes in order to avoid interferences. The sink produce the scheduling, it is sent to the nodes and executed.

PEDAMACS~\cite{ergen06} also uses such a scheme but with only one radio channel. Nodes have different transmission powers. The sink can reach all the nodes in the network, the other nodes have two transmission powers: one to communicate and one to identify their interferers. The protocol needs a global synchronization of the network. This is achieved thanks to synchronization packets that are sent by the sink to the whole network. The protocol consists in three phases. In the first one, the  topology learning phase, each node learns its interferers and neighbors by sending hello packets in contention periods. During the second phase, the topology collection phase, the information is sent to the sink using a contention mechanism. A schedule is computed by the sink and sent to the nodes. The method used to produce the schedule is to linearize the graph of the network (containing the interference edges) and to give the same color to non interfering levels. The slots are allocated to non-interfering sets of nodes with the same color. During the third phase the nodes communicate in their allocated slots.

A drawback of centralized protocols is that the sink needs to retrieve information on the topology which is not scalable and can lead to high energy consumption and memory issues. Further in this paper we study the differences between PEDAMACS and our proposition mechanisms and the trade-offs those differences underline. We choose PEDAMACS because to the best of our knowledge it is the only protocol that integrates MAC and routing layers being able to guarantee real-time delivery of packets with the hypothesis that there is only one channel available and taking into account energy issues. Moreover, it gives us the opportunity to compare centralized and localized real-time protocols.

We also compare our solution with a nondeterministic solution: XMAC with a gradient routing scheme. It allows to highlight the better reliability of deterministic solutions even with probabilistic radio links.

In the previous paragraphs we present different solutions to schedule the access to the channel and route packets in WSNs. Nevertheless, none defines both MAC and routing layers that can ensure bounded end-to-end delays while being localized, thus more scalable, and energy efficient. Our solution is based on a deterministic and localized access scheme and on greedy and opportunistic routing mechanisms.

%------------------------------------------------------------------------- 
\section{Hypothesis and problem statement}
\subsection{Hypothesis}
\label{hyp}
In this section we discuss the assumptions we consider. Assumptions we make are mainly related to the sensor capabilities, the radio environment and the application. Moreover, our proposition is based on more specific requirements.

Assumptions on the limited capacities of sensor nodes:
\begin{itemize}
 \item Sensors have a limited amount of energy;
 \item The nodes have a limited amount of memory.
\end{itemize}

Assumptions linked to the radio:
\begin{itemize}
 \item The radio is half-duplex and mono-channel;
 \item We assume a 2-hop interference model, meaning that nodes are able to receive packets from their 1-hop neighbors and to detect activity of their 2-hop neighbors. We will discuss this assumption in the performance evaluation section (section \ref{perfs});
 \item Radio links are symmetric.
\end{itemize}

Assumption linked to the application:
\begin{itemize}
 \item The traffic intensity is low and consists in alarm packets converging toward the sink;
\end{itemize}

Assumption more specific to our propositions:
\begin{itemize}
 \item Nodes have local coordinates which give information on the number of hops from the sink and which are unique in a 2-hop neighborhood. This provided by the mechanisms of \cite{Mouradian12}.
\end{itemize}

\subsection{Problem statement}
We want to design a protocol in order to support real-time alarm gathering in WSNs. The characteristics of WSNs impose to this protocol to respect requirements. It should be scalable due to the large scale of WSNs, it should be reliable because the applications are critical and the wireless links are unreliable. The protocol should ensure that the end-to-end delay is lower than a given bound. In the remainder of the paper we propose and evaluate RTXP, a protocol which respects those requirements.

%------------------------------------------------------------------------- 

\section{Proposition: a novel Real-Time X-layer Protocol, RTXP}
In this section we detail RTXP, a cross-layer (MAC and routing) protocol which guarantees a bounded end-to-end delay for alarm packets. We first give the general ideas of the protocol, we describe the virtual coordinate it uses as a metric, we go further into the mechanisms of the protocol and we discuss supported applications in function of the parameters of the radio.

\subsection{General idea}
\label{general}
As energy is a main concern in WSNs, RTXP uses a duty-cycle mechanism. We call \emph{awake period} the period in which the nodes are awake and \emph{sleep period} the one in which they turn off their radio.

\begin{figure}[h!]
  \centering
  \includegraphics[width=4in, keepaspectratio=true]{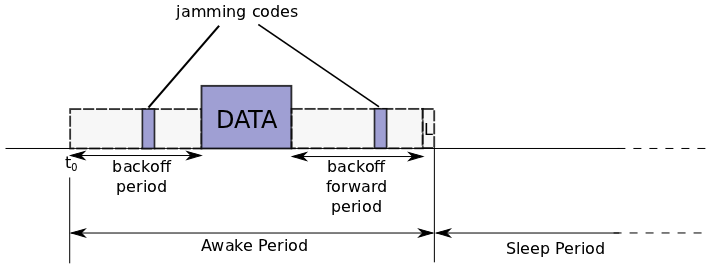}
  \caption{Description of the proposition}
  \label{protocolGeneral}
\end{figure}

As depicted on Figure~\ref{protocolGeneral}, the \emph{awake period} is divided into three main phases:
\begin{itemize}
 \item A backoff period in which nodes, with a packet in their queue, contend to reserve the channel. The contention occurs between nodes of a 2-hop neighborhood. Each node has a backoff timer that is calculated from its coordinate and is unique in the 2-hop neighborhood. During the backoff the nodes sense the channel. Then, either the node's backoff expires and the node sends a jamming code meaning that it gains access to the channel, or it detects a jamming code before the end of its backoff timer (meaning that it loses the contention).
 \item  A time slot during which the data packet is transmitted.
 \item Another backoff period during which all the nodes that received the data packet contend to forward it. This backoff period works the same way as the one for channel reservation.
\end{itemize}

During the L slot, any node that lost the contention for the channel can send a jamming code which triggers a new awake period for the nodes that detects it. As it will be discussed in section \ref{vcs} and thanks to the uniqueness of the coordinate in a 2-hop neighborhood, the access to the channel and the selection of the forwarder are deterministic. 

\subsection{Wakeup time: preambles versus synchronization}
\label{sync}
In the previous subsection we assume that all the nodes of a given 2-hop neighborhood are awake at the beginning of the first backoff period noted $t_{0}$. There are two ways of achieving this goal. 

The first is by using a long preamble as in \cite{polastre04}. With this technique a node that has a packet to send, sends a preamble which duration is equal to the duration of a duty-cycle. A node, when it wakes up, senses the channel, if it detects energy it stays awake until the end of the preamble. In the case of RTXP, the end of the preamble corresponds to $t_{0}$. This solution does not need to maintain a global synchronization of the nodes. Nevertheless, the emission of the preamble consumes energy and time. Indeed, when an alarm converge toward the sink, each relaying node must send the preamble before executing the three phases presented in previous subsection.

The second way is to have a global synchronization of the nodes of the network so that all the nodes can wake up at the same time. This solution may seem energy costly but recent work \cite{lampin13} showed that with a suitable synchronizing technique the global synchronization can be more energy efficient. Moreover, with global synchronization a packet can be forwarded several times during a duty-cyle. For these reasons we choose to use a global synchronization for RTXP. So for the remainder of this paper we assume that the nodes of the WSN are synchronized. Nevertheless, we can notice that the two schemes allow to implement real-time protocols.
  
\subsection{Virtual coordinate system}
\label{vcs}
The Virtual Coordinate System (VCS) used by RTXP consists in a 1-D coordinate, which has been proposed in~\cite{Mouradian12}, and is based on two parameters. The first one is the hop-count to the destination node, but many nodes can have the same hop-count. The nodes having the same hop-count can be seen, conceptually, as forming concentric rings centered on the sink. So the second information is the connectivity of a node with the different rings (noted $\mathit{offset}$). These two parameters are then integrated into one coordinate which classifies the nodes of a same ring. This classification is related to the connectivity, nodes having more neighbors in proportion in the rings nearer to the sink are classified before nodes having more neighbors in proportion in the rings further from the sink. This information allows to give priority to nodes more connected to lower rings for the forwarders selection. Figure~\ref{coord} illustrates the coordinate, $R$ is the estimated radio range of a node, $\mathit{offset}$ embed the information on connectivity (c.f. ~\cite{Mouradian12} for more details) and $n$ is a ring number. The coordinate is given by $
Coord=(n-1)*R+\mathit{offset} \mbox{ with } \mathit{offset}<R$. In the proposed solution the probability of having two nodes with the same coordinate in a 2-hop neighborhood is low but not null (this issue is discussed in \cite{Mouradian12}). In this paper we assume that the coordinate is unique in a 2-hop neighborhood. 
\begin{figure}[h]
  \centering
  \includegraphics[width=2.5in, keepaspectratio=true]{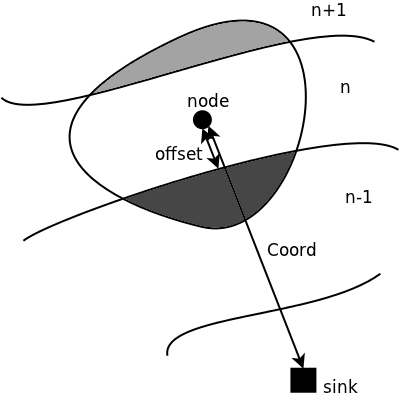}
  \caption{Conceptual view of the 1-D coordinate used by RTXP}
  \label{coord}
\end{figure}

\subsection{In-depth detail of RTXP}
\label{rtxp-det}
First, we describe further the three phases of the protocol mentioned in section \ref{general}.
 
\textbf{Phase 1.} In the first backoff period, each node being awake and having a packet to send contends for the channel. During the contention a node senses the channel. If it detects energy on the channel before the end of its backoff timer, it loses the contention. Otherwise, it sends a jamming code. A jamming code is a short sequence of bits, possibly random. The technique is very similar to preamble sampling \cite{polastre04} but with the jamming code being shorter than a typical preamble. If a node loses the contention it can notify it in a dedicated slot (noted L for Lost on Figure \ref{protocolGeneral}) by sending a jamming code in it. Every node that receives a jamming code in the L slot will stay awake for another awake period. As mentioned previously, we assume that the nodes are able to detect jamming codes from their 2-hop neighborhood in order to prevent the hidden terminal problem. The backoff timer is calculated with a bijective function from the coordinate (for example the $\mathit{offset}$ is translated directly into milliseconds so the function is of the type $y=x$). The lemma 4.1 ensures that there is no collision in a two-hop neighborhood.

\begin{lemma}
 If the coordinates are unique in a 2-hop neighborhood and the backoff function is bijective then there is only one node that wins the contention.
\end{lemma}

\begin{proof}
  \label{lemma_contention}
 We do a proof by contradiction. Let's suppose there are two nodes that win the contention (i.e. there is a collision). That means they have the same backoff time which implies that either they have the same coordinate or the backoff function is not bijective which is a contradiction.
\end{proof}.

\begin{figure}[h!]
  \centering
  \includegraphics[width=4in, keepaspectratio=true]{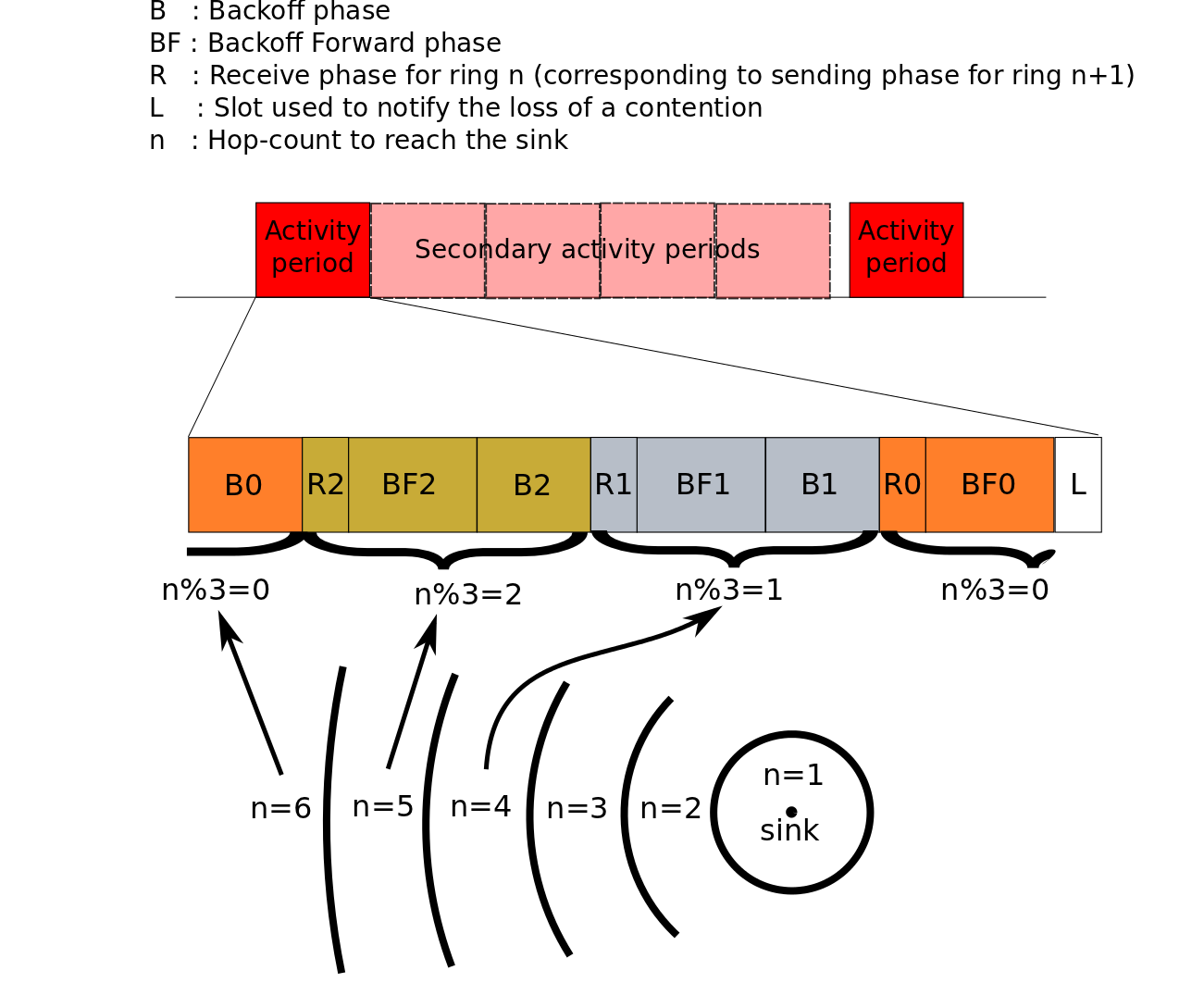}
  \caption{Description of the proposition}
  \label{protocol}
\end{figure}

\textbf{Phase 2.} During the second phase (data emission and reception), the node (with hop-count $n$) who won the contention of the first phase sends its packet and the nodes in range, with hop-count $n-1$, receive it.

\textbf{Phase 3.} The third phase is another contention period (backoff forward phase). All the nodes that received a packet in the second phase contend to know which one will forward it. As in the first phase, the backoff function is bijective and calculated from the coordinate. We can notice, in this case, that we want to preserve the order given by the coordinate, thus the function must also be strictly monotonic. The first node which backoff ends, sends a jamming code to notify the others that it will be the forwarder. We can notice that this mechanism is a kind of opportunistic forwarding based on the coordinate.

\textbf{Organization of the phases.} As said in section \ref{sync}, we choose to use a global synchronization scheme. This allow to forward a packet several times during one duty-cycle because potential forwarders are already synchronized. We also mentioned that we assume a 2-hop interference model, meaning that nodes that are three hops away can transmit at the same time. Thus, to give a chance to each node to transmit during a duty cycle whatever its hop-count is, we should define three \emph{awake periods} per duty-cycle (but we keep only one $L$ slot), one for nodes with $3j$ hop-count, one for $3j+1$ and one for $3j+2$ with $j \in \mathbb{N}$. As depicted on Figure~\ref{protocol} an \emph{activity period} is composed of three \emph{awake periods}. A packet can, at most, reach a node with $3j$ hop-count from a node with a $3j+3$ hop-count during one \emph{activity period}.

Figure~\ref{protocol} depicts the different phases for nodes with different hop-counts. $B_i$, $R_i$ and $BF_i$ correspond respectively to backoff, receive and backoff forward phases
with $i=n \bmod{3}$. For example a node 6 hops away from the sink contends in $B_0$ if it has a packet to send. It sends the packet in $R_2$ if it wins the contention. It wakes up in $R_0$ to potentially receive a packet and if it has received one it does the $BF_0$ phase to try to forward the packet. 

When a node, which has a packet to transmit, loses the contention during the backoff phase, it has the opportunity to claim a new \emph{activity period} (named \emph{secondary activity period}). This new \emph{activity period} follows the previous one, without sleeping time. Only nodes which sense a jamming code in the L slot stay active.

In WSNs links are unreliable, nodes may experience fading and shadowing. Thus packets may not be correctly received. In order to mitigate the impact of unreliable link we use a broadcast-based routing scheme, a data packet is received by several nodes. The forwarder is elected during the BF phase. Moreover, the jamming code sent during the election of the forwarder (BF phase) is used as an acknowledgement. A node which sends a packet in the R phase then waits for a jamming code, if it does not receive one, the packet is considered lost and the node send a jamming in the L slot to request a new \emph{activity period}. The packet is resent in the new \emph{activity period}. 

\textbf{Example.} The figure \ref{example} shows an example with a simple network. The nodes A and B have both data to send to the sink at the beginning. They contend in the first part of the \emph{activity period}. B wins the contention so it can send its packet in the $R$ phase to C. Similarly, C sends it to D, and D then forward it to the sink. At the end of the first \emph{activity period}, node A sends a jamming code in the L slot because it just lost the contention previously. Because we assume a 2-hop interference model, B, C and D sense the jamming code so they stay awake for a new \emph{activity period}. In this second activity, only A has a packet to send thus A wins the contention and the packet is forwarded to C. At the end of the second \emph{activity period} all the nodes go to sleep mode because no node transmits during the L slot (every packet has done at least one hop).

\begin{figure}[!h]
  \centering
  \includegraphics[width=4in, keepaspectratio=true]{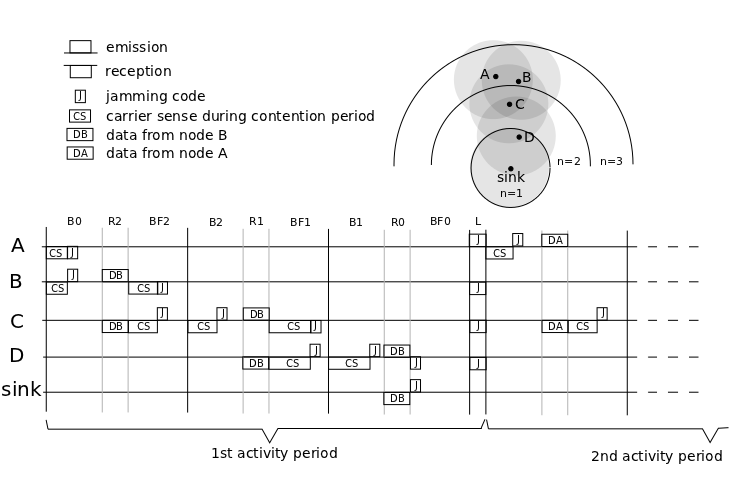}
  \caption{Example considering 4 nodes where nodes A and B have a packet to transmit}
  \label{example}
\end{figure}

\begin{table}[!h]
	\centering
	\begin{tabular}{|l|l|r|}
  		\hline
 	 	symbol & signification \\
  		\hline
		$D_{B}$ & Duration of the backoff ($B$) phase\\
		$D_{BF}$ & Duration of the backoff forward ($BF$) phase\\
		$D_{R}$ & Duration of the receive ($R$) phase\\
		$D_{L}$ & Duration of the L slot \\
		$D_{jamming}$ & Duration of the jamming code \\
		$D_{activity\_period}$ & Duration of an \emph{activity period} \\
  		$D_{awake}$ & Duration of awake period\\& during an \emph{activity period}  \\
		$D_{sleep}$ & Duration of the sleep period\\
		$WCTT_{RTXP}$ & Theoretical bound on the end-to-end delay for RTXP \\
  		$DC$ & The duty-cycle ratio \\
		$C_{RTXP}$ & Capacity of RTXP \\
		$NB_{hop\_max}$ & The maximum number of hops from the sink\\
		$E_{backoff}$ & Energy consumed during the backoff ($B$) phase\\
		$E_{backoff\_forward}$ & Energy consumed during the backoff forward $BF$ phase\\
		$E_{TX\_jamming}$ & Energy consumed during the emission of a jamming code\\
		$E_{TX\_packet}$ & Energy consumed during the emission of a packet\\
		$E_{RX\_packet}$ & Energy consumed during the reception of a packet\\
		$E_{1hop\_RTXP}$ & Energy consumed by RTXP to do one hop\\
%		$k$ & Number of neighbors of a node\\
  		\hline
	\end{tabular}
	\caption{Notations used in the description of RTXP}
	\label{notations}
\end {table}

%Moreover the radio must be able to sample the channel, with a period that is smaller than the difference between any two backoff durations. If it is not the case, a collision can occur because a node may not sense the beginning of a jamming code before it emits its own jamming code.

%The dimensioning of these periods highly depends on the characteristics of the radio chip used for the deployment, the higher the bitrate, the shorter the periods. 

\section{Theoretical analysis}
\label{theory}
\subsection{Delay, capacity and energy}
In order to compute a bound on the end-to-end delay, called Worst Case Traversal Time (WCTT), we propose to calculate the worst duration for one hop and multiply it by the maximum number of hops in the network. We start by defining intermediate delays. The notations used in this section are detailed in Table~\ref{notations}.

$B$ and $BF$ durations depend on the backoff duration which is function of the offset ($backoff=f(\mathit{offset})$). Indeed, each node must have the possibility to send a jamming code during the period so the duration of the backoff phases must be equal to the maximum backoff duration plus the duration of a jamming code: \begin{eqnarray}D_{B}=D_{BF}= max(backoff)+D_{jamming}\end{eqnarray} 

The R phase duration is the time required to transmit a data packet (noted $D_{R}$). In our case the data packet is an alarm packet which size is in the order of magnitude of a few dozens of bytes.

The duration of the L slot ($D_{L}$) is equal to the duration of a jamming code ($D_{jamming}$).

The sleep period duration ($D_{sleep}$)  is calculated based on the duration of an \emph{awake period} ($D_{awake})$. This duration corresponds to the $B$ phase plus the $BF$ phase plus two $R$ phases (one to send data to lower hop-count neighbors and one to receive data from upper hop-counts nodes). The calculation of $D_{sleep}$ depends on the duty-cycle ratio ($DC$) (which depends itself on the application):
\begin{eqnarray}
DC=D_{awake}/(D_{sleep}+D_{awake})
\end{eqnarray}
\begin{eqnarray}
D_{awake}=D_{B}+D_{BF}+2\times D_{R}+D_{L}
\end{eqnarray}
\begin{eqnarray}
D_{sleep}=D_{awake}\times (\frac{1}{DC}-1)
\end{eqnarray}

The duty-cycle ratio typically depends on the application characteristics, this aspect is discussed in section \ref{to}. 

The \emph{activity period} is represented on the figure \ref{protocol} its duration is given by:
\begin{eqnarray}
D_{activity\_period}&=&3 \times (D_{B}+D_{BF} +D_{R}) + D_{L}
\end{eqnarray}

The worst case duration for one hop is given by $D_{activity\_period}+D_{sleep}$ because a packet, in the worst case, is transmitted after having lost the contention every time until the limit given by the end of the sleep period. We take the maximum number of hops plus one because there is a delay (which is at most one duty cycle period) between the instant the event is sensed and the first emission of the corresponding packet.
\begin{eqnarray}
WCTT_{RTXP}=(NB_{hop\_max}+1)\times (D_{activity\_period}+D_{sleep}) 
\end{eqnarray}

The number of \emph{activity periods} is limited by the length of the sleep period (which depends itself on the duty-cycle value). We define the capacity of RTXP ($C_{RTXP}$) as the number of packets that can be transmitted into a 2-hop neighborhood during the duty-cycle. Given the duty-cycle duration $(D_{activity\_period}+D_{sleep})$, the capacity is:
\begin{eqnarray}
 C_{RTXP}= \lfloor \frac{D_{activity\_period}+D_{sleep}}{D_{activity\_period}} \rfloor 
\end{eqnarray}

%attention ici on prend l'hyppothèse du lien parfait !!!!! il faut dire plus tôt que le slot L sert aussi à dire qu'un noeud n'a pas reçut d'aquitement pour son paquet
\begin{theorem}
 Let $n\in \mathbb{N}$ and $p\in \mathbb{N}$ be respectively a hop-count number and the number of packets in a 2-hop neighborhood at hop-count $n$. Assuming that packet can only be lost because of collisions, all packets in every 2-hop neighborhood at hop-count $n$ will reach hop-count $n-1$ in at most a duty cycle period if $p<C_{RTXP}$.
\end{theorem}

\begin{proof}
 We do a proof by contradiction. Let's suppose one packet did not reach hop-count $n-1$ in one duty-cycle period. Then either the packet was lost or it was delayed until the end of the period. As we assumed the only way to lose a packet is because of a collision. By lemma 4.1 we know that it is not possible so it is a contradiction. If the packet is delayed until the end of the duty-cycle period that means the node lost every contention until the end for that packet. So it lost more than $\frac{D_{activity\_period}+D_{sleep}}{D_{activity\_period}}=C_{RTXP}$ contentions, so there were more than $C_{RTXP}$ packets ($p>C_{RTXP}$) in a 2-hop neighborhood which is a contradiction.
\end{proof}

Thus, under this capacity limit, the delivery ratio is 100\% with the hypothesis that packet loss is only due to interferences with other nodes. As we mentioned previously, it is not the case in practice, because nodes may experience fading or shadowing. We also mentioned that this issue is mitigated by broadcast routing and retransmissions. This issue is discussed in the performance evaluation section (section \ref{perfs}). 

%We can notice that  a packet can do at most three hops during one activity period if the nodes, which forward it, do not lose any contention. If one node loses a contention, the packet may do only one hop because the node has to notify the loss in the L slot, and then, the $n-2$ and $n-3$ rings may not be awake anymore.

The energy used during one hop if the relaying node wins the contention is:
\begin{eqnarray}
\nonumber E_{1hop\_RTXP}&=&E_{backoff}+E_{TX\_jamming}+E_{TX\_packet}\\
& & +E_{RX\_packet}+E_{backoff\_forward}\\
\nonumber & & +E_{TX\_jamming}
\end{eqnarray}

\subsection{Trade-offs}
\label{to}
The capacity depends on the inverse of the duty cycle ratio, so the longer the sleep period the higher the capacity. 
%Moreover with a longer sleep period more energy is saved. 
Nevertheless the bound on the delay of a packet increases with the sleep period. Thus a trade-off which depends on the application has to be found between the bound on the delay and the capacity of the protocol. In the case of applications with low traffic and short time constraints a small sleep period should be used. In the case of applications with high traffic and less tight time constraints a longer sleep period should be preferred.

\begin{table}[!h]
	\centering
	\begin{tabular}{|l|l|r|}
  		\hline
 	 	Parameter & Value \\
  		\hline
		Maximum number of hops & 5\\
		Duration of jamming code & 200$\mu$s \\
		Duration of Backoff phases (B and BF) & 10.2ms\\
		Duration of data transmission (R phase)& 32ms\\
		Duty-cycle ratio & from 100\% to 1\% of activity\\
  		\hline
	\end{tabular}
	\caption{Parameters used for the plot of capacity vs WCTT}
	\label{val}
\end {table}

\begin{figure}[!h]
  \centering
  \includegraphics[width=3in, keepaspectratio=true]{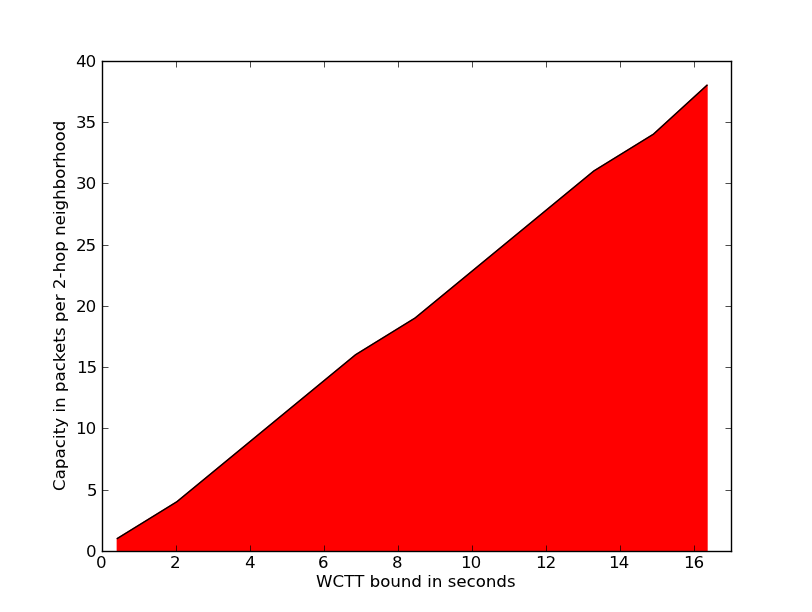}
  \caption{Capacity of RTXP in function of the WCTT}
  \label{capvsdelay}
\end{figure}

Figure \ref{capvsdelay} is a plot of the capacity given in packets per 2-hop neighborhood that can be handled during a duty cycle in function of the WCTT (which depends itself on the duty-cycle ratio). The colored part corresponds to feasible zone for RTXP. The expression is derived from equations 6 and 7. For example with these values (given in Table \ref{val}), if the application requires a WCTT of 6 seconds the maximum capacity of RTXP is 15. This means that 15 packets can be transmitted in a 2-hop neighborhood during a duty-cycle.

\section{Performances evaluation and protocols comparisons}
\label{perfs}
In this section we evaluate the performances of our solution by simulation and compare it to state of the art protocols. We compare RTXP with a centralized real-time solution, PEDAMACS. With unreliable links, it is not possible to give hard real-time bound on the end-to-end delay. Indeed it is not possible to give with certainty the number of retransmissions needed for a packet to be correctly received. We thus compare RTXP with a nondeterministic solution, to show that deterministic solutions are more reliable.

\subsection{PEDAMACS}
To the best of our knowledge RTXP is the first localized real-time X-layer protocol for WSNs. Existing real-time X-layer solutions are centralized. We, thus, compare our solution to a centralized mono-channel solution, PEDAMACS. As described in section \ref{centralized}, in the PEDAMACS protocol, the sink node produces a scheduling frame after retrieving topology information (tree graph). In this section we define the worst case traversal time and the energy consumption of PEDAMACS.

In \cite{coleri02}, authors state that it is ensured that all the packets reach the sink during the scheduling phase (i.e. during the scheduling frame). The maximum length of the scheduling frame depends on the topology, some possible cases are given in~\cite{coleri02}. We will consider the case of a general tree graph $G=(V,E)$ with a 2-hop interference model. Such a graph is retrieved by the sink during the initialization phases as described in section \ref{centralized}. In this case, the maximum frame length is:
\begin{eqnarray} 
  WCTT_{PEDAMACS}=3\times(|V|-1)\times T_{slot} 
\end{eqnarray}

%test with special cases ?
Equation 9 shows that, in the case of PEDAMACS, the bound on end-to-end delay (the worst case) does not depend on the number of hops but on the number of nodes (the worst case is a linear network).

We evaluate the energy-consumption induced by a packet to do one hop. In the case of PEDAMACS it is only the energy used by one node to send the packet and by another to receive it:
\begin{eqnarray}
E_{1hop\_PEDAMACS}=E_{TX\_packet}+E_{RX\_packet}
\end{eqnarray}

\subsection{Simulation environment and parameters}
The simlulations are performed with the WSNet simulator \cite{wsnet}. WSNet is a discrete event simulator which is designed especially for the simulation of WSN characteristics. 
For the simulations, we generated 140 random topologies, where nodes are distributed on a 50x50 units plane according to a uniform law. The 140 topologies are divided into sets of 20 topologies of $m\times100$ nodes with $m$ an integer $\in [2,8]$. A simulation is run for each topology.

\begin{table}[!h] %We did 20 simulations for each point (
	\centering
	\begin{tabular}{|l|l|r|}
  		\hline
 	 	Parameter & value \\
  		\hline
		Number of nodes & 100 to 800 \\
  		Bitrate & 500kbps \\
  		Radio range & 10 units \\
  		Area & 50$\times$50 units \\
  		Packet size & 100 bytes \\
		Jamming code duration & 200 $\mu$s\\
  		Backoff duration & 10,2 ms \\
		Duty-cycle ratio & 1\%\\
  		Path loss exponent & 2\\
		$\sigma$ of log-normal law & 4\\
		
  		\hline
	\end{tabular}
	\caption{Simulation parameters}
	\label{param}
\end {table}
%provide a graph with reception probability in function of the distance for FS and LN

During each simulation 200 packets are sent. The traffic consists in alarm generated periodically from a random point of the network. Every period a node is picked randomly among all nodes of the network to be the origin of the alarm. In the simulations we considered two rates, 1 alarm every 5 seconds and 1 alarm every seconds so we can observe how the protocols simulated under different traffic loads react. For the topologies simulated, these rates are far from the capacity limit expressed in equation 7. Indeed, we use a duty-cycle ratio of 1\%, from equations 2 and 4 we can deduce that the duration of a duty cycle is about 2.5 seconds. According to equation 7, it means that 100 packets can be forwarded in 2.5 seconds in a 2-hop neighborhood (about 40 packets per seconds). With 1 alarm every 5 seconds and 1 alarm every seconds, we are thus in cases which correspond to the low traffic hypothesis made in section \ref{hyp}. It allows the nodes to sleep most of the time (few secondary activity periods triggered).

The radio model is a mono-channel half-duplex with a 500kbps rate. Parameters of the simulations are detailed in table \ref{param}. 

We use two propagation models, the free-space and the log-normal shadowing models. The free-space model allows us to evaluate the performances of our protocol when the packet losses are only due to nodes interfering each others. It allows us to confront the statement made in section \ref{protocol} with simulation results. The log-normal shadowing model provide a much more realistic propagation model for WSNs \cite{Zuniga04}.

\subsection{RTXP vs PEDAMACS}
In this section we present the results of simulations of RTXP and PEDEMACS on free-space and log-normal shadowing propagation channels. We compare the end-to-end delays of the packets, the energy spent during simulations and the delivery ratio of these protocols. We compare the delays observed during the simulations with the theoretical WCTTs of RTXP and PEDAMACS, respectively expressed equations 6 and 9. On the figures the end-to-end delay of a packet is represented by a cross, we choose to represent all the values to be able to observe the distribution of the delays. The circles correspond to the average delay for a given number of neighbors. The black solid curve corresponds to the theoretical WCTT.

\subsubsection{Free Space propagation model.}
\begin{figure}[!h]
    \subfigure[1 packet every 5 seconds]{\includegraphics[width=2.4in, keepaspectratio=true]{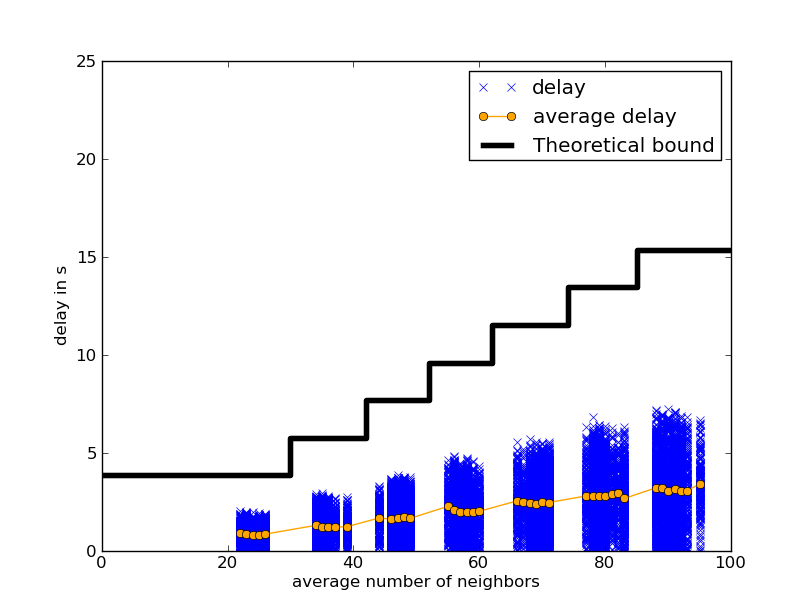}
\label{delay-p5s}}
    \hfil
\subfigure[1 packet per second]{\includegraphics[width=2.4in, keepaspectratio=true]{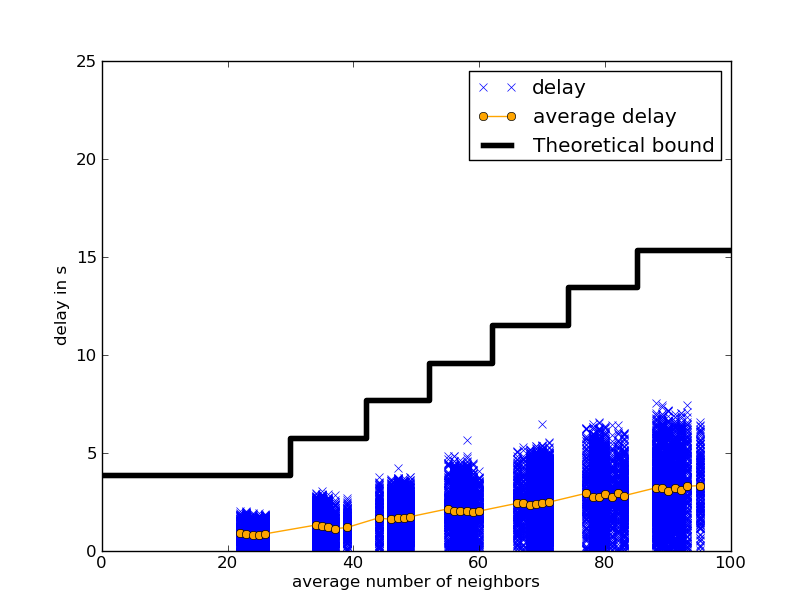}
      \label{delay-p1s}}
\caption{PEDAMACS - free-space propagation model}
\end{figure}

Figures \ref{delay-p5s} and \ref{delay-p1s} respectively represent the end-to-end delay of alarms in function of the average number of neighbors for PEDAMACS for 1 packet every 5 second and 1 packet per second rates. First we can notice that all the packets meet their deadlines. This is ensured by the global scheduling. Moreover, the scheduling also ensured that there are no interfering nodes are communicating at the same time. Since it is the only way to lose a packet with free-space propagation model, we observe a delivery ratio of 100\%. The delay is not affected by the traffic load because the scheduling frame do not change according to it.

\begin{figure}[!h]
    \subfigure[1 packet every 5 seconds]{\includegraphics[width=2.4in, keepaspectratio=true]{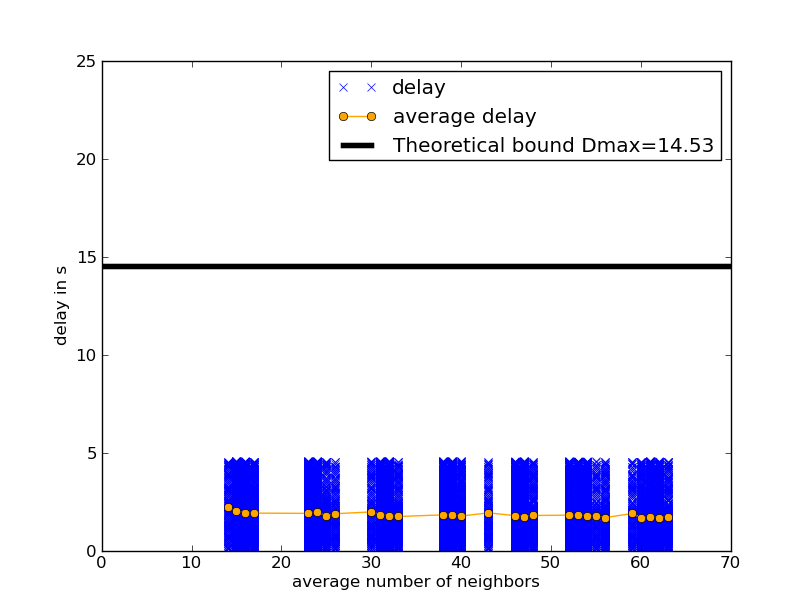}
\label{delay-r5s}}
    \hfil
\subfigure[1 packet per second]{\includegraphics[width=2.4in, keepaspectratio=true]{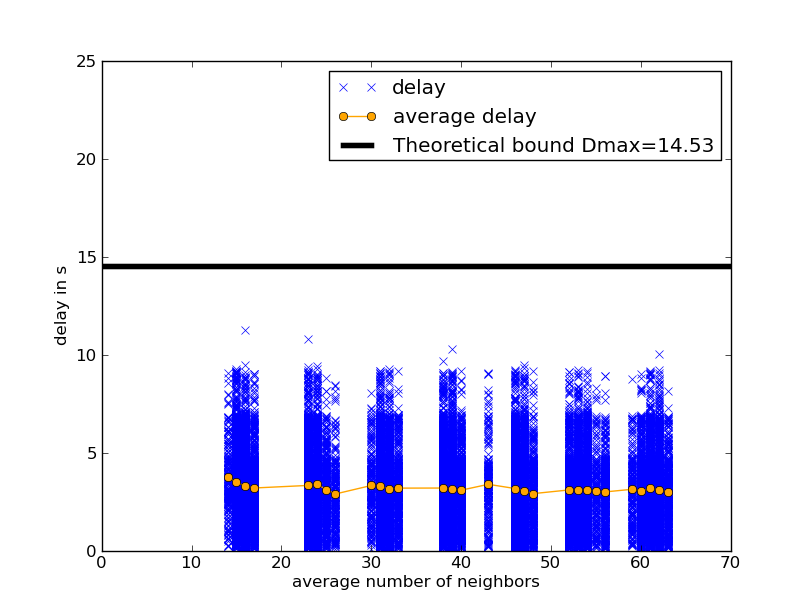}
      \label{delay-r1s}}
\caption{RTXP - free-space propagation model}
\end{figure}

Figures \ref{delay-r5s} and \ref{delay-r1s} depict the end-to-end delays and theoretical bound for RTXP. In this case, we also observe that all the packets meet their deadlines as predicted in section \ref{theory}. Moreover, the delivery ratio is also 100\%. Nevertheless, the increase in the load affects the delay, it produces an increase of the delays of some packets. This is due to the fact that when the load increases, it triggers more \emph{secondary activity periods} because there are more packets which are at the same time in the same 2-hop neighborhood. On the other hand, the delay does not vary with the average number of neighbors.

\textbf{Energy consumption.}
\begin{figure}[!h]
  \centering
  \includegraphics[width=3in, keepaspectratio=true]{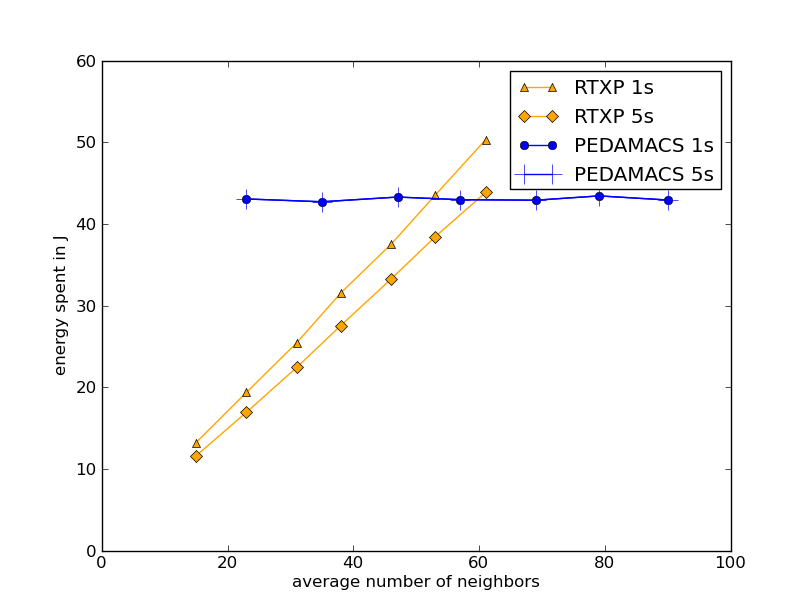}
  \caption{Maximum energy consumption of runtime of PEDAMACS and RTXP}
  \label{energy}
\end{figure}

Figure \ref{energy} depicts the maximum energy consumption observed. Each point of the curves corresponds to the maximum value for 20 topologies of the same size. The energy calculation for PEDAMACS and RTXP is done respectively according equations 10 and 8. 

The energy spent by PEDAMACS grows really slightly with the average number of neighbors because the number of nodes increases. Nevertheless the growth is not very important because the number of hops in the network does not change and the number of packets transmitted remains the same (200 alarms are produced). When the load increases it does not affect the energy spent by PEDAMACS because the scheduling remains the same.

The energy spent by RTXP grows linearly with the average number of neighbors. This is due to the fact that data packet are broadcasted to the neighbors of the sender. With the 1 packet per second rate the energy spent is higher than with the 1 packet every 5 seconds rate. The higher the alarm rate is the more \emph{secondary activity periods} are triggered.

PEDAMACS has a higher energy consumption than RTXP for networks with an average number of neighbors below 50. This is due to the fact that with PEDAMACS the nodes are waking up even if there is no traffic as a result of the scheduling. PEDAMACS is thus more suited to a periodic traffic where all nodes have a packet to send during each scheduling fame than to an alarm traffic. RTXP, on the contrary, adapts to the traffic load. If there is no alarms nodes sleep most of the time, if there are lots of alarms, secondary periods are triggered to handle the traffic.

\subsubsection{Log-normal shadowing propagation model.}

\begin{figure}[!h]
    \subfigure[delay]{\includegraphics[width=2.4in, keepaspectratio=true]{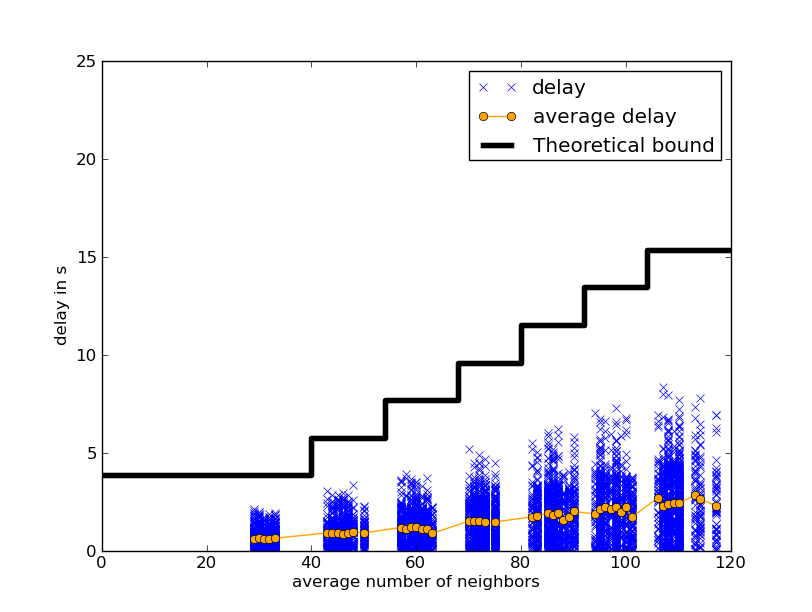}
\label{delay-pln}}
    \hfil
\subfigure[delivery ratio]{\includegraphics[width=2.4in, keepaspectratio=true]{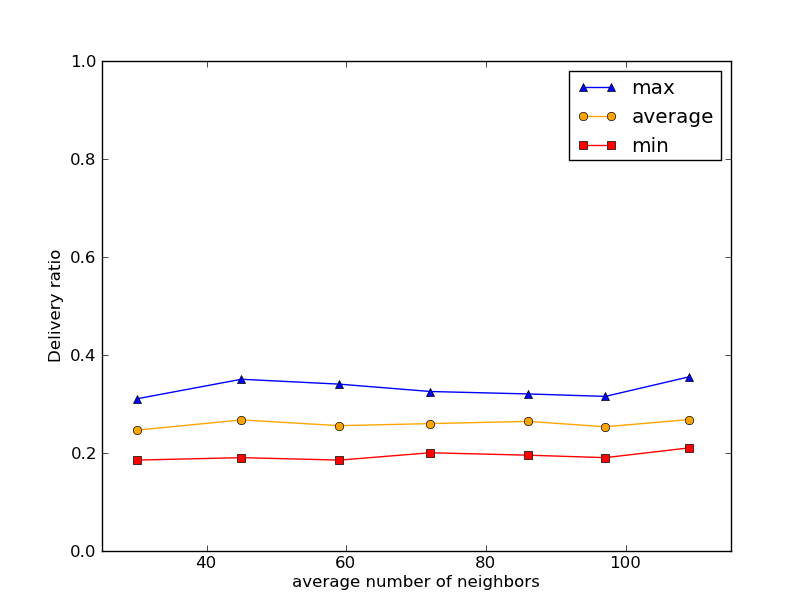}
      \label{dr-p}}
\caption{PEDAMACS - log-normal shadowing}
\end{figure}

Figure \ref{delay-pln} depicts the delay of the packets and the theoretical bound for PEDAMACS in the case of the log-normal shadowing model. In this case as well, no packet miss its deadline. Nevertheless, the Figure \ref{dr-p} represents the minimum, maximum, and average delivery ratios observed during the simulations. The values are very low, most of the packets are lost because of the harsh channel conditions. PEDAMACS does not implement any retransmission mechanism so if a transmission fails the packet is irremediably lost.

\begin{figure}[!h]
    \subfigure[delay]{\includegraphics[width=2.4in, keepaspectratio=true]{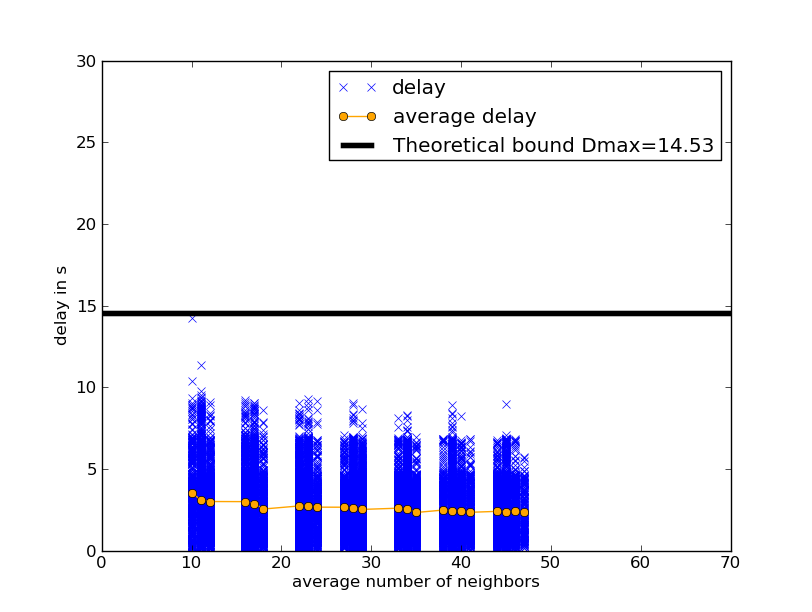}
\label{delay-rln-noret}}
    \hfil
\subfigure[delivery ratio]{\includegraphics[width=2.4in, keepaspectratio=true]{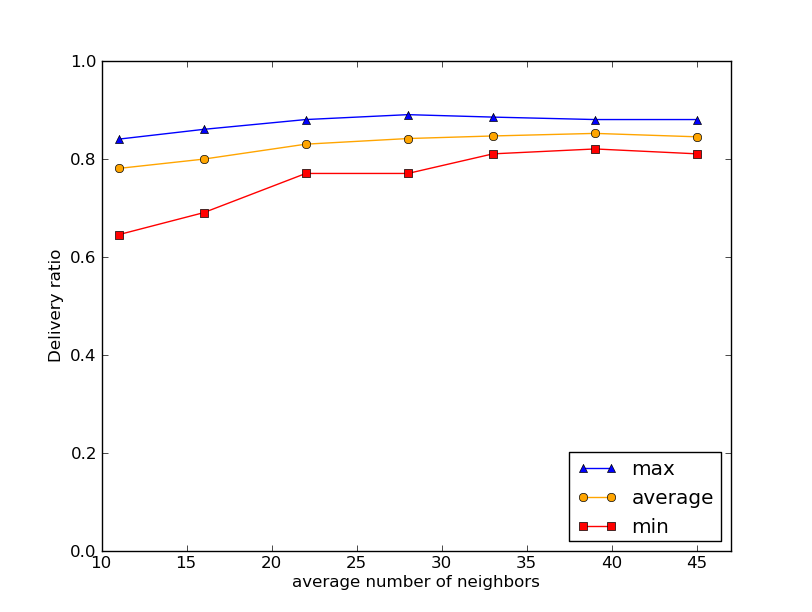}
      \label{dr-r-noret}}
\caption{RTXP - log-normal shadowing without retransmissions}
\end{figure}

Figure \ref{delay-rln-noret} depicts the delay of the packets and the theoretical bound for RTXP in the case of the log-normal shadowing model with no retransmission mechanism. As in the case of PEDAMACS, all the packets meet the deadline. But in this case the delivery ratio (Figure \ref{dr-r-noret}) is higher. This is due to the broadcast based scheme implemented by RTXP which increases the reliability as described in section \ref{rtxp-det}.

\begin{figure}[!h]
    \subfigure[delay]{\includegraphics[width=2.4in, keepaspectratio=true]{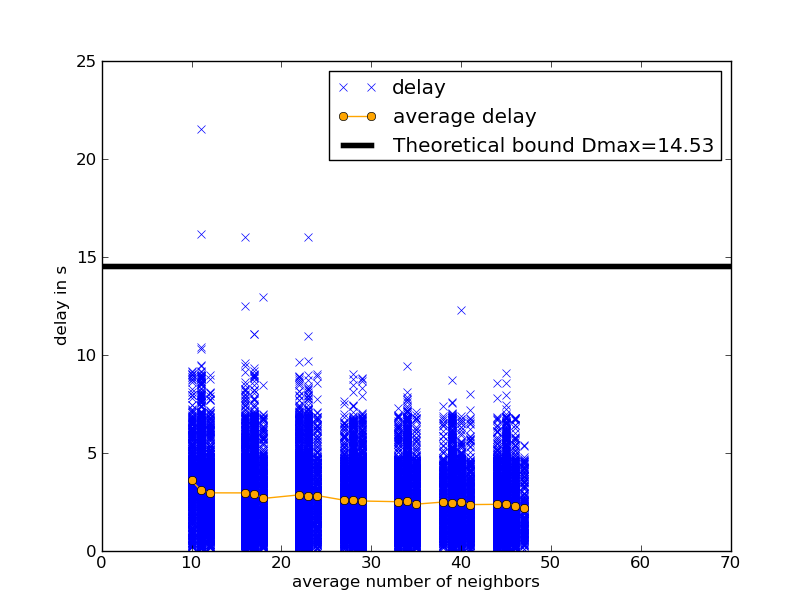}
\label{delay-rln}}
    \hfil
\subfigure[delivery ratio]{\includegraphics[width=2.4in, keepaspectratio=true]{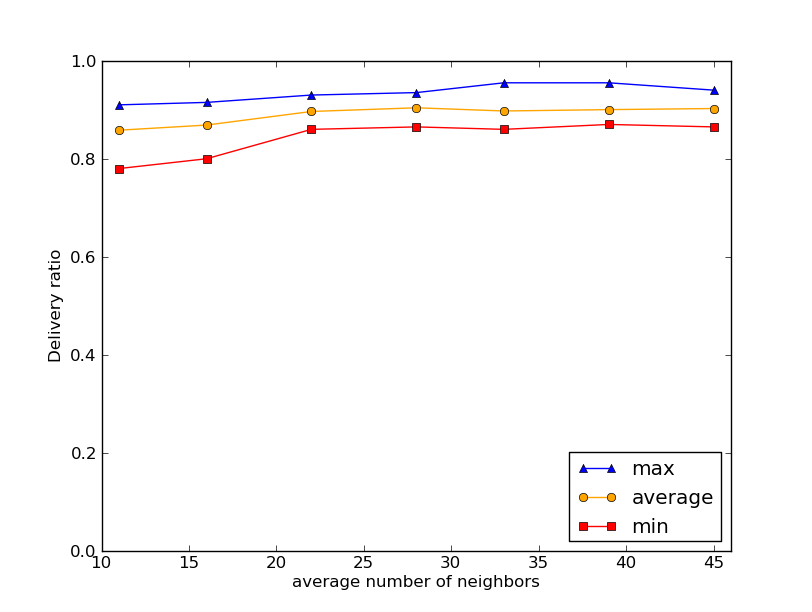}
      \label{dr-r}}
\caption{RTXP - log-normal shadowing with retransmissions}
\end{figure}

Figure \ref{delay-rln} depicts the delay of the packets and the theoretical bound for RTXP in the case of the log-normal shadowing model with retransmissions. In this case we notice that few packets miss the deadline. The retransmission mechanism is described in section \ref{rtxp-det}, if a sender does not detect a jamming code during the Backoff Forward phase, it sends a jamming code in the L slot to trigger a \emph{secondary activity period}. It then retransmits the packet in the new \emph{activity period} (in our implementation a packet can be retransmitted 5 times per duty-cycle, then it is resent in the next duty-cycle). In some cases, there are too many retransmissions so the packet cannot meet the deadline. Nevertheless, the retransmission mechanism allows to have a higher delivery ratio even under harsh channel conditions. Figure \ref{dr-p} represents the minimum and maximum delivery ratios observed during the simulations. The values are higher than those observed in the case of PEDAMACS and it improves the results of RTXP without retransmissions.

Under harsh radio channel conditions it is not possible to ensure that all the packets are received. Neither it is possible to ensure that all packets are received before the deadline. This is due to the probabilistic nature of the radio link, indeed there is a chance that a packet is not correctly received even after many retransmissions. RTXP is designed with the goal of avoiding probabilistic behaviors, channel access and forwarder selection are deterministic, so the behavior is predictable and we can ensure that packets meet their deadline. Nevertheless, the radio channel introduces a probabilistic aspect thus one can legitimately ask if it is worth it to have a deterministic protocol on a probabilistic channel. In the next sections this issue is further investigated. 

\subsection{Comparison with a non real-time solution}
In this section we compare RTXP to a non real-time solution under harsh radio channel conditions in order to verify that having deterministic behavior in the protocol actually improve the real-time performances. We choose to compare RTXP to a XMAC and gradient routing solution \cite{ye04}. XMAC \cite{buettner06} is a preamble MAC protocol as described in section \ref{mac}, but it does not wake up all the neighbors of the sender. The preamble is composed of shorts packets and response slots. Nodes alternately sleep and wake up. When a node wakes up it senses the channel, if it receives a preamble packet and is the destination of the packet it answers in a response slot, otherwise it goes back to sleep. In our case, we use an opportunistic gradient routing scheme, meaning that any node that receives a preamble packet and has a smaller hop-count than the sender can answer and become the forwarder of the current packet. A node that has a packet to send, first senses the channel, if the channel is free it transmits the preamble packets. If it senses activity it backs off for a random duration and retry then. The access to the channel is thus not deterministic.

With the XMAC and gradient protocol, the end-to-end delay depends on the number of hops a packet has to do to reach the sink. A packet has to wait for a preamble length to do one hop (at most one duty-cycle period). In order to fairly compare this solution with RTXP, we take duty-cycle duration of one third of the duty-cycle of RTXP (Because in RTXP a packet can do up to three hops during a single duty-cycle). This choice actually disadvantage RTXP because XMAC preamble lasts half a duty cycle on average. We use equation 6 as the theoretical bound for XMAC with gradient. XMAC define an acknowledgment packet and the number of retransmission can be specified. During the simulations different values are tested in order to monitor the effect retransmissions have on reliability and delay.

The alarm rate used in the simulations is 1 alarm every 5 seconds. The channel model is log-normal propagation model. Figures \ref{no-retry}, \ref{5-retry} and \ref{500-retry} respectively depict the results for 0, 5 and 500 retransmissions. The same parameters as previously are monitored: end-to-end delay and delivery ratio.

\begin{figure}[!h]
    \subfigure[delay]{\includegraphics[width=2.4in, keepaspectratio=true]{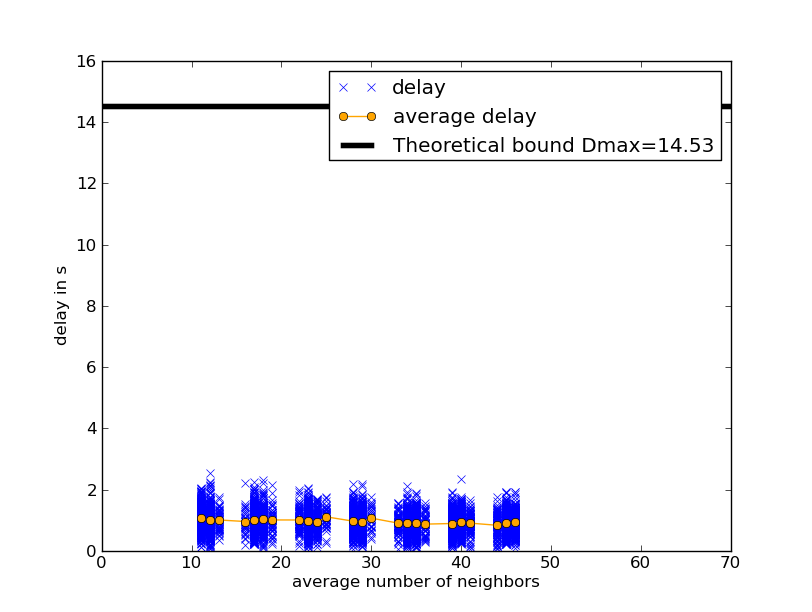}
\label{delay-x0r}}
    \hfil
\subfigure[delivery ratio]{\includegraphics[width=2.4in, keepaspectratio=true]{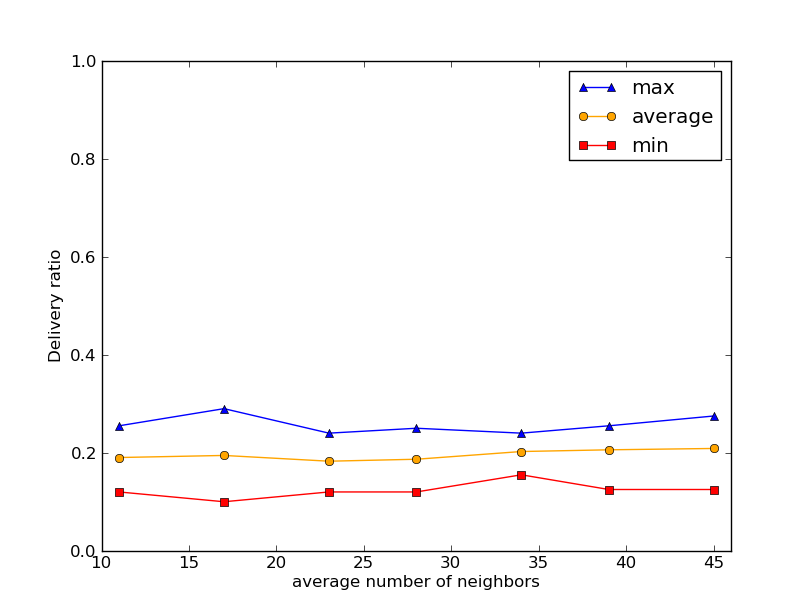}
      \label{dr-x0r}}
\caption{XMAC with gradient: no retransmission}
  \label{no-retry}
\end{figure}

Figure \ref{delay-x0r} shows that every packet, which arrives to the sink, respects the deadline in the case there is no retransmission. Nevertheless, the delivery ratio, shown on Figure \ref{dr-x0r}, is very low compared to the one achieved with RTXP (as shonw on Figure \ref{dr-r-noret}).

\begin{figure}[!h]
    \subfigure[delay]{\includegraphics[width=2.4in, keepaspectratio=true]{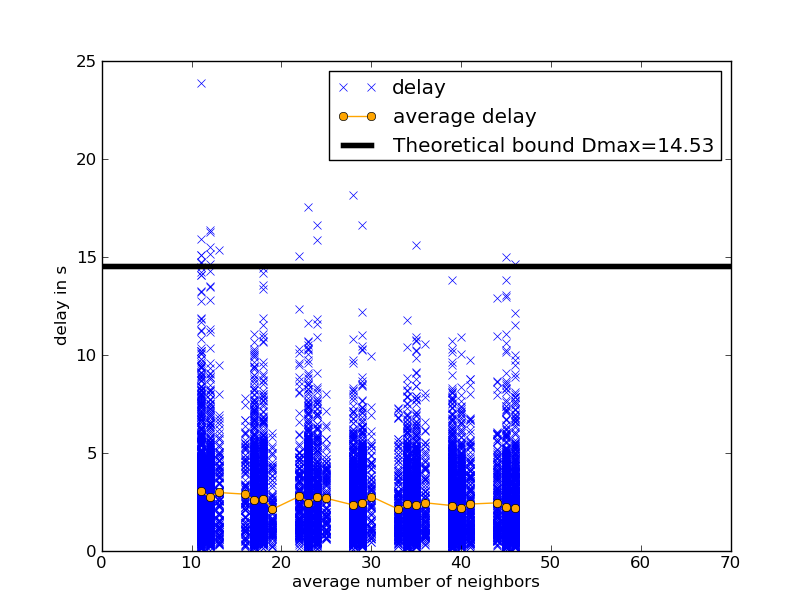}
\label{delay-x5r}}
    \hfil
\subfigure[delivery ratio]{\includegraphics[width=2.4in, keepaspectratio=true]{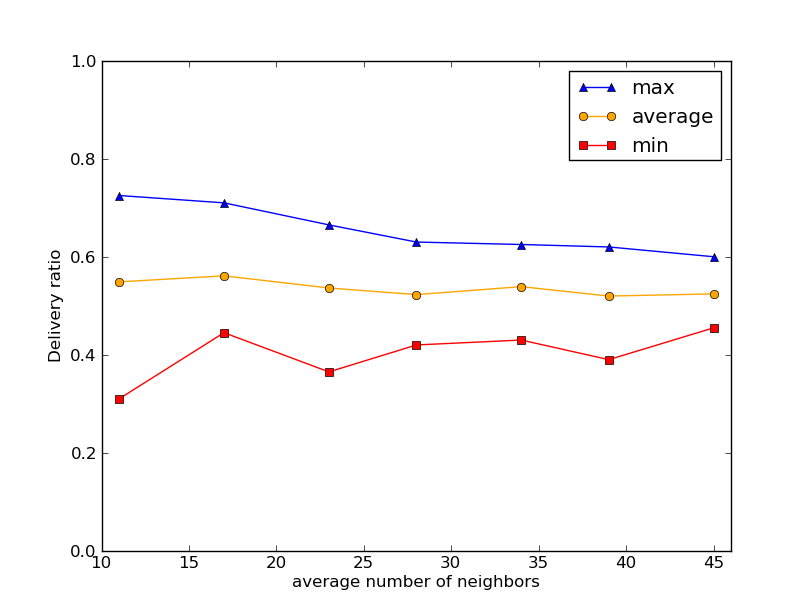}
      \label{dr-x5r}}
\caption{XMAC with gradient: 5 retry}
\label{5-retry}
\end{figure}

In the case of the 5 retransmissions setting, Figure \ref{delay-x5r} shows that some packets miss the deadline. It occurs because the retransmissions increase the end-to-end delay. The delivery ratio, depicted on Figure \ref{dr-x5r} is higher than in the previous case. Nevertheless, the amount of packets that miss the deadline is higher than in the case of RTXP as it can be seen on Figure \ref{delay-rln}. The delivery ratio, represented on Figure \ref{dr-x5r}, is higher than in the previous case because packets have more probabilities to be successfully transmitted when the number of retransmissions increases. Nevertheless, the delivery ratio is still smaller than with RTXP as can be seen on Figure \ref{dr-r}. Moreover the difference between maximum and minimum values of delivery ratio is smaller in the case of RTXP, it is thus more stable. 

\begin{figure}[!h]
    \subfigure[delay]{\includegraphics[width=2.4in, keepaspectratio=true]{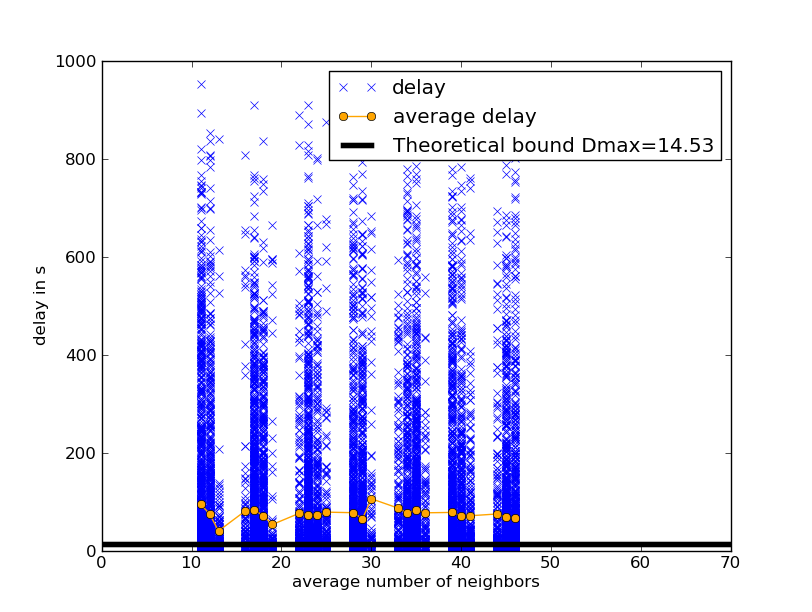}
\label{delay-x500r}}
    \hfil
\subfigure[delivery ratio]{\includegraphics[width=2.4in, keepaspectratio=true]{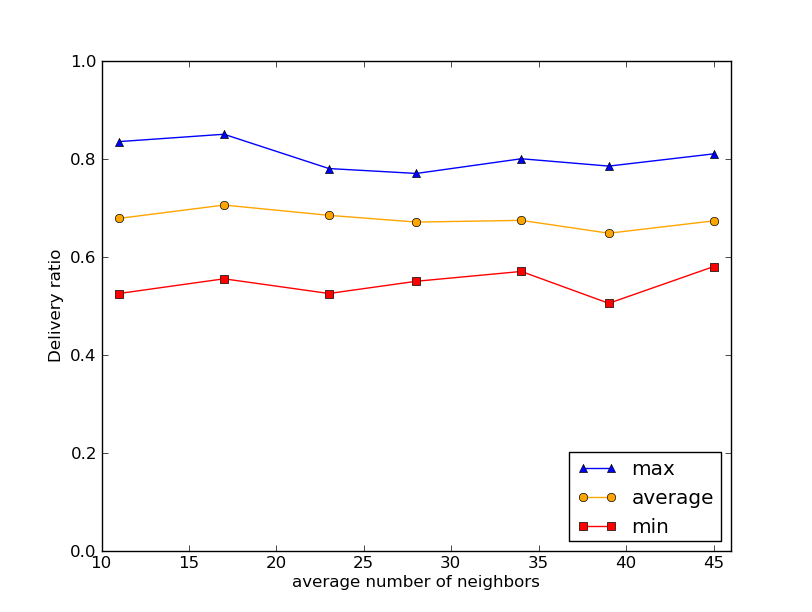}
      \label{dr-x500r}}
\caption{XMAC with gradient: 500 retry}
\label{500-retry}
\end{figure}

In our implementation of RTXP a packet is retransmitted 5 times during one duty-cycle. If it still has not be correctly received, it will be retransmitted during the next duty cycle. This means that in the case of RTXP there is no bound on the number of times a packet can be retransmitted. Thus in order to fairly compare RTXP with XMAC gradient, we choose a very high number of retransmission: 500. As depicted on Figure \ref{delay-x500r}, most of the packets miss the deadline with delays up to several hundreds of seconds. Nevertheless, as shown on Figure \ref{dr-x500r}, it results in a slight increase of the delivery ratio, but it remains below the values of RTXP as can be seen on Figure \ref{dr-r}. These high delays are mostly due to the fact that the high number of retransmissions induces a high occupation of the channel resulting in longer delays to access the channel.

These results show that introducing determinism for channel access and routing leads to better performances even with a probabilistic radio link.

%------------------------------------------------------------------------- 
\section{Conclusion and future work}
In this paper we present RTXP, a solution to handle real-time alarms in WSNs. To the best of our knowledge RTXP is the first localized MAC and Routing protocol for WSN able to guarantee end-to-end delays. We describe the proposition and give its theoretical bound on the end-to-end delay and its real-time capacity. By simulation we compare RTXP and PEDAMACS, a centralized scheduling solution. We show that RTXP is more suited to alarm traffic than PEDAMACS. By simulating the protocols under harsh radio channel conditions we show that it is not possible to give hard guarantees on the delay under unreliable radio link assumptions. Nevertheless by favorably comparing RTXP to a non real-time solution, we demonstrate the usefulness of real-time approaches even with unreliable links.

In this paper we derive the theoretical delay bound and capacity from general statements made in the protocol description. From these statements we also construct simple proofs of properties of RTXP. Nevertheless, to be trusted the protocol must be described in a formal language and verified using a formal verification technique. A future work will be to apply model checking techniques to RTXP. Finally experimentation on real sensors has to be performed in order to verify the performances of our solution.



%% References


\bibliographystyle{plain}
\bibliography{RR}
\newpage
\tableofcontents

\end{document}